\documentclass[12pt]{article}
\usepackage{graphicx,amssymb,amsmath,amsthm,url,bm,verbatim}

\newtheorem{theorem}{Theorem}
\newtheorem{lemma}[theorem]{Lemma}

\newcommand{\bR}{\mathbb{R}}

\newcommand{\be}{\mathbf{e}}

\newcommand{\bs}{\mathbf{s}}
\newcommand{\bu}{\mathbf{u}}
\newcommand{\bv}{\mathbf{v}}
\newcommand{\bo}{\mathbf{1}}
\newcommand{\br}{\mathbf{r}}
\newcommand{\bx}{\mathbf{x}}
\newcommand{\by}{\mathbf{y}}

\newcommand{\bnu}{{\bm \nu}}
\newcommand{\bh}{\mathbf{\hat{x}}}

\newcommand{\mmu}{\overline{\mathbf{u}}}
\newcommand{\mv}{\overline{\mathbf{v}}}
\newcommand{\mr}{\overline{\mathbf{r}}}
\newcommand{\mx}{\overline{\mathbf{x}}}
\newcommand{\mh}{\overline{\mathbf{\hat{x}}}}

\newcommand{\sh}{\sigma_{\bh}}
\newcommand{\su}{\sigma_{\bu}}
\newcommand{\sv}{\sigma_{\bv}}
\newcommand{\sr}{\sigma_{\br}}
\newcommand{\sx}{\sigma_{\bx}}

\begin{document}
\title{Maximum likelihood decoding for multilevel channels with gain and offset mismatch}
\author{Simon R.\ Blackburn\\
Department of Mathematics\\
Royal Holloway University of London\\
Egham, Surrey TW20 0EX\\
United Kingdom}
\maketitle

\begin{abstract}
K.A.S. Immink and J.H. Weber recently defined and studied a channel with both gain and offset mismatch, modelling the behaviour of charge-leakage in flash memory. They proposed a decoding measure for this channel based on minimising Pearson distance (a notion from cluster analysis). The paper derives a formula for maximum likelihood decoding for this channel, and also defines and justifies a notion of minimum distance of a code in this context.
\end{abstract}

\section{Introduction}

We begin by defining some notation. Let $n$ be an integer, $n\geq 3$. All our vectors will have length $n$, and will have entries in the real numbers $\mathbb{R}$. For a vector $\bx$, we write $x_i$ for the $i$th entry of $\bx$, we write
\[
\mx=\frac{1}{n}\sum_{i=1}^n x_i
\]
for the mean of $\bx$ and we write
\[
\sx=\sqrt{\sum_{i=1}^n(x_i-\mx)^2}
\]
for the (unnormalised) standard deviation of $\bx$. We write $\bo$ for the all-one vector of length $n$, and call any scalar multiple of $\bo$ a constant vector. For vectors $\bu$ and $\bv$ that are not constant vectors, the {\em Pearson correlation coefficient} $\rho_{\bu,\bv}$ is defined by
\[
\rho_{\bu,\bv}=\frac{\sum_{i=1}^n (u_i-\mmu)(v_i-\mv)}{\su\sv}.
\]
Finally, the \emph{Pearson distance} $\delta_\mathrm{Pearson}(\bu,\bv)$ between vectors $\bu$ and $\bv$ is defined to be
\[
\delta_\mathrm{Pearson}(\bu,\bv)=1-\rho_{\bu,\bv}.
\]
Since $\rho_{\bu,\bv}$ lies between $-1$ and $1$, the Pearson distance lies between $0$ and~$2$. Both Pearson distance and Pearson correlation are well-known concepts in the area of cluster analysis.

The channel considered by Kees A. Schouhamer Immink and Jos H. Weber~\cite{ImminkWeber14} is defined as follows. If the vector $\bx$ is sent through the channel, the channel outputs the received vector $\br$ where
\[
\br=a(\bx+\bnu)+b\bo.
\]
Here $a$ (the \emph{gain}) and $b$ (the \emph{offset}) are unknown real numbers, with $a>0$, and
\[
\bnu=(\nu_1,\nu_2,\ldots,\nu_n)
\]
where the $\nu_i$ are independently normally distributed with mean $0$ and standard deviation $\sigma$.

The channel is motivated by the properties of flash memory. We give some basic details of this setting here; see~\cite{BezCamerlenghi03,SalaImminkDolecek15} for more detailed introductions, and see (for example)~\cite{JiangMateescu09,JiangSchwartz08,JiangSchwartz10} for another approach to modelling the problem using rank modulation codes. Flash memory is made up of an array of floating-gate transistors, known as flash cells. Data is stored in each cell by varying the charge (equivalently, the voltage) on the cell. In single level cell (SLC) flash memory, each cell stores one bit of information depending on whether the voltage level is zero or non-zero. In more recent multi-level cell (MLC) systems, more information is stored by allowing the cell to be charged at one of several discrete non-zero voltage levels. The vector $\bx$ corresponds to the voltages we wish to store in a block of $n$ cells, so $x_i$ is the voltage we wish to store in the $i$th cell. We cannot hope to initialise a cell with the exact voltage we wish: the errors in this process give rise to the error term $\bnu$. Over time, the voltage in each cell drops due to charge leakage. We assume that the function that gives this voltage change is unknown, but is affine and is independent of which cell in the block we are examining. The unknown coefficients $a$ and $b$ specify this function; the coefficient $a$ is positive since charge leakage is monotonic increasing: the more charge we have initially, the more we have after leakage. The received vector $\br$ thus models the set of voltages we retrieve from a block of cells we have initialised with voltages corresponding to $\bx$.

We note that the channel does not model some aspects of flash memory: intercell coupling (where the charge on one cell influences the charge on neighbouring cells) is not modelled in any way; nor is the possibility that the magnitude of the error in the charging process depends on the charge in some way. Nevertheless, the channel is very natural and captures key properties of the process of retrieving data from flash memory. 

Immink and Weber assume that the vectors $\bx$ lie in some finite subset $C$ of $\mathbb{R}^n$. (In fact, they assume that $C\subseteq\{0,1,\ldots, q-1\}^n$ for some fixed integer $q$.) This corresponds to the fact that we initialise each cell with one of a finite discrete set of voltages. To ensure unique decoding in the absence of noise, they assume that if $\bx\in C$ then no other codeword $\by\in C$ has the form $\by=a\bx+b\bo$ for real numbers $a$ and $b$ with $a$ positive. They also assume that no constant vector lies in $C$. This makes the Pearson distance between any pair of vectors in $C$ well-defined; see Section~\ref{sec:comments} for additional motivation for this assumption. Weber, Immink and Blackburn~\cite{ImminkWeberPearson} have studied maximal codes $C\subseteq\{0,1,\ldots,q-1\}^n$ with these properties.

A decoder based on Pearson distance is proposed in this setting in~\cite{ImminkWeber14}. So we decode a received vector $\br$  as $\bh$, where $\bh\in C$ minimises $\delta_\mathrm{Pearson}(\br,\bh)$. One motivation for this choice is that Pearson distance behaves well with respect to an affine charge-leakage function, since
\[
\delta_\mathrm{Pearson}(\br,\bh)=\delta_\mathrm{Pearson}(a\br+b\bo,a\bh+b\bo).
\]
Pearson distance has a natural geometric meaning: see Section~\ref{sec:comments} for a brief discussion.

In this paper, we derive a maximum likelihood decoding function for the channel in~\cite{ImminkWeber14}, and compare a decoder based on this function with a decoder based on minimising Pearson distance. We also propose and justify a notion of minimum distance for codes used with this channel. 

We should emphasise that the model makes no assumptions on the distribution of the unknown (`nuisance') parameters $a$ and $b$: if we know something about these distributions, other decoding methods might be appropriate. For example, if $a$ is known to be very close to $1$, then decoding based on minimising Euclidean distance is sensible; Immink and Weber~\cite{ImminkWeber15} have proposed a decoder based on minimising a weighted sum of Euclidean and Pearson distances in some situations.

The remainder of the paper is structured as follows. Section~\ref{sec:preliminaries} sets up notation, and contains some preliminary lemmas. In Section~\ref{sec:MLD} we show how to achieve Maximum Likelihood Decoding for this channel. Pearson distance is not the measure to use for Maximum Likelihood decoding, but is often a good approximation to it: simulations show comparable performance between both MLD and Pearson decoders. Section~\ref{sec:distance} defines and justifies a minimum distance measure for codes designed for the channel. In Section~\ref{sec:comparison}, we give some results of simulations that compare the approach in~\cite{ImminkWeber14} with the one taken here. Finally, Section~\ref{sec:comments} provides some comments on various aspects of the model in~\cite{ImminkWeber14}.

\section{Preliminaries}
\label{sec:preliminaries}

This section contains notation that will be used in the remainder of this paper. Some simple facts, which will often be used without further comment, are also stated.

We define $||\bu||$ to be the Euclidean length of $\bu\in\bR^n$, and we define $\delta(\bu,\bv)$ to be the Euclidean distance between $\bu,\bv\in\bR^n$.

Define the subspace $Z$ of $\bR^n$ by
\begin{align*}
Z&=\{\bx\in\bR^n:\mx=0\}\\
&=\left\{(x_1,x_2,\ldots,x_n)\in\bR^n:\sum_{i=1}^nx_i=0\right\}.
\end{align*}
Let $\zeta:\bR^n\rightarrow Z$ be defined by
\[
\zeta(\bx)=\bx-\mx\bo.
\]
We can think of $\zeta$ as a `normalisation', applying an offset to a vector so that it has mean zero. Using $\zeta$ allows the formulas given in the introduction to be expressed in a more geometric way. We now give more details. We see that
\begin{equation}
\label{eqn:sigma_geometric}
\sigma_\bu=||\zeta(\bu)||.
\end{equation}
We write $\langle \bx,\by\rangle$ for the standard inner product (the dot product) of $\bx$ and $\by$. So 
\[
\langle \bx,\by\rangle=\sum_{i=1}^nx_iy_i.
\]
Since $\langle \bx,\by\rangle=||\bx||\,||\by||\cos \theta$ where $\theta$ is the angle between $\bx$ and $\by$, we see that
\begin{align}
\label{eqn:rho_inner_product}
\rho_{\bu,\bv}&=\frac{\langle \zeta(\bu),\zeta(\bv)\rangle}{\sigma_{\zeta(\bu)}\sigma_{\zeta(\bv)}}\\
&=\frac{||\zeta(\bu)||\,||\zeta(\bv)||\cos \theta}{||\zeta(\bu)||\,||\zeta(\bv)||}\nonumber\\
\label{eqn:rho_geometric}
&=\cos\theta,
\end{align}
where $\theta$ is the angle between $\zeta(\bu)$ and $\zeta(\bv)$.

Finally, we note that $\zeta(\zeta(\bu))=\zeta(\bu)$, that $\zeta(\bu+\bv)=\zeta(\bu)+\zeta(\bv)$ and that $\sigma_{\zeta(\bu)}=\sigma_{\bu}$.

\section{Maximum likelihood decoding}
\label{sec:MLD}

This section provides a proof of the following theorem:

\begin{theorem}
\label{thm:MLD}
A maximum likelihood decoder decodes a received vector $\br$ to the codeword $\bh$ which minimises $\ell_\br(\bh)$, where
\begin{equation}
\label{eqn:MLD}
\ell_{\br}(\bh)=\begin{cases}
\sigma_\bh^2(1-\rho^2_{\br,\bh})&\text{ when }\rho_{\br,\bh}> 0,\\
\sigma_\bh^2&\text{ otherwise.}
\end{cases}
\end{equation}
\end{theorem}

Before proving this theorem, we provide a geometrical interpretation for the formula~\eqref{eqn:MLD}. For a non-zero vector $\br\in\bR^n$, define
\begin{align*}
U_\br&=\{a'\br+b'\bo\mid a',b'\in\bR\},\text{ and}\\
U^+_\br&=\{a'\br+b'\bo\mid a',b'\in\bR, a'>0\}.
\end{align*}
So $U_\br$ is a subspace, and $U^+_{\br}$ is a half-subspace, of $\bR^n$. For a vector $\br\in\bR^n$ we write $R_\br$ for the ray from the origin in the direction of $\br$, so
\[
R_\br=\{a'\br:a'\in\bR,a'>0\}.
\]

\begin{lemma}
\label{lem:MLD}
Let $\br$ and $\bh$ be vectors in $\bR^n$. Let $d_1$ be the 
Euclidean distance between $\bh$ and $U^+_\br$. Let $d_2$ be
the Euclidean distance between $\zeta(\bh)$ and $R_{\zeta(\br)}$. Then
\[
d_1^2=d_2^2=\ell_{\br}(\bh).
\]
\end{lemma}

\begin{proof} We start by proving that $d_1=d_2$.  Let $\bu=a'\zeta(\br)=a'\br-a'\mr\bo\in R_{\zeta(\br)}$. Then
\[
\delta(\zeta(\bh),\bu)=\delta(\bh-\mh\bo,\bu)=\delta(\bh,\bu+\mh\bo),
\]
and $\bu+\mh\bo=a'\br+(\mh-\mr)\bo\in U^+_\br$. So $d_1\leq d_2$. 

Let $\bu=a'\br+b'\bo=a'\zeta(\br)+(b'+\mr)\bo\in U^+_\br$. Then
\[
\delta(\bh,\bu)=\delta(\zeta(\bh),\bu-\mh\bo)=\delta(\zeta(\bh),a'\zeta(\br)+(b'+\mr-\mh)\bo)\geq \delta(\zeta(\bh),a'\zeta(\br)),
\]
since $\zeta(\bh),a'\zeta(\br)\in Z$ and since $\bo$ is orthogonal to $Z$. Since $a'\zeta(\br)\in R_{\zeta(\br)}$, we see that $d_2\leq d_1$. Hence $d_1=d_2$.

We now prove that $d_2^2=\ell_{\br}(\bh)$. There are two cases, depending on whether or not the closest point $P$ to $\zeta(\bh)$ on the line generated by $\zeta(\br)$ lies in the ray $R_{\zeta(\br)}$: see Figure~\ref{fig:dist_to_ray}. The first case, when $P$ lies on the ray, happens if and only if $\langle\zeta(\bh),\zeta(\br)\rangle>0$. This happens exactly when $\rho_{\bh,\br}>0$, by~\eqref{eqn:rho_inner_product}. In this case,
\[
d_2^2=||\zeta(\bh)||^2\sin^2\theta=||\zeta(\bh)||^2(1-\cos^2\theta)=\sigma_{\bh}^2(1-\rho^2_{\br,\bh}),
\]
where $\theta$ is the angle between $\zeta(\bh)$ and $\zeta(\br)$, by~\eqref{eqn:sigma_geometric} and~\eqref{eqn:rho_geometric}. In the second case, when $\langle\zeta(\bh),\zeta(\br)\rangle\leq 0$, the distance between $\zeta(\bh)$ and the ray $R_{\zeta(\br)}$ is given by the distance from $\zeta(\bh)$ to the origin. So
\[
d_2^2=||\zeta(\bh)||^2=\sigma_{\bh}^2,
\]
by~\eqref{eqn:sigma_geometric}. This establishes the lemma.
\begin{figure}
\begin{center}
\includegraphics[width=100mm]{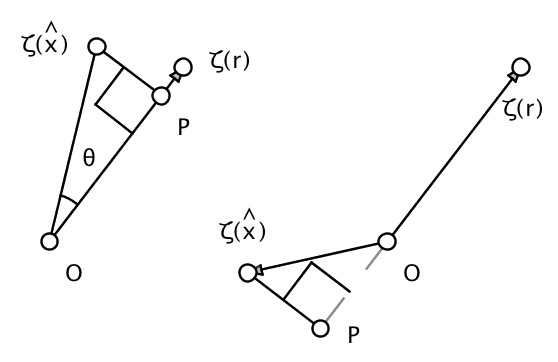}
\end{center}
\caption{The distance of a point to a ray: two cases}
\label{fig:dist_to_ray}
\end{figure}
\end{proof}

\begin{proof}[Proof of Theorem~\ref{thm:MLD}]
Since the components of $\bnu$ are picked independently according to a normal distribution with mean $0$ and standard deviation $\sigma$, each value of $\bnu$ is associated with the value of the corresponding normal Probability Density Function $f(\bnu)$, where
\[
f(\bnu)=\prod_{i=1}^n\frac{1}{\sigma\sqrt{2}\pi}\exp(-\nu_i^2/(2\sigma^2)).
\]

For vectors $\br$ and $\bh$, define
\[
L_{a,b}(\bh\mid\br)=f((\br-b\bo)/a-\bh).
\]
This is the likelihood of $\bh$ given $\br$ when $a$ and $b$ are fixed, since $\bnu=(\br-b\bo)/a-\bh$ in this case.

In maximum likelihood decoding, we decode a received vector $\br=a(\bx+\bnu)+b\bo$ as $\bh\in C$, where $\bh$ is the codeword that maximises
\begin{align*}
\max_{a,b\in \mathbb{R}, a>0} L_{a,b}(\bh\mid\br)&=
\max_{a,b\in \mathbb{R}, a>0} f((\br-b\bo)/a-\bh)\\
&= \max_{a',b'\in\mathbb{R},a'>0} f(a'\br+b'\bo-\bh),
\end{align*}
where $a'=1/a$ and $b'=b/a$. The logarithm function is strictly increasing on the positive real numbers, and $f$ is a positive function. So equivalently we want to find $\bh\in C$ that maximises $\max_{a',b'\in\mathbb{R},a'>0}\log f(a'\br+b'\bo-\bh)$. But
\[
\log f(a'\br+b'\bo-\bh)=-n\log(\sigma\sqrt{2}\pi)-\frac{1}{2\sigma^2}\sum_{i=1}^n(a'r_i+b'-\hat{x}_i)^2.
\]
Since $-n\log(\sigma\sqrt{2}\pi)$ is a constant (in other words, independent of $\bh$ and $\br$), and since $\frac{1}{2\sigma^2}$ is a positive constant, we see that a maximum likelihood decoder finds a codeword $\bh$ that \emph{minimises}
\[
\min_{a',b'\in\bR,a'>0}\sum_{i=1}^n(a'r_i+b'-\hat{x}_i)^2,
\]
which is the square of the Euclidean distance between $U_\br^+$ and $\bh$. But, 
by Lemma~\ref{lem:MLD}, this is exactly the same as minimising the function $\ell_\br(\bh)$, as required.
\end{proof}

We describe techniques to reduce the amount of computation the maximum likelihood decoder needs. Firstly, the value $\sigma_{\bh}^2$ can be precomputed for all codewords $\bh\in C$. Secondly, for codes such as $2$-constrained codes~\cite{ImminkWeber14} that are preserved under permuting their cooordinates, we can significantly reduce the number of codewords we need to consider by making the following observations. The value of $\sigma_{\bh}$ is not changed if we permute the coordinates of $\bh$, and the value of $\rho_{\br,\bh}$ is maximised when we permute the coordinates of $\bh$ to have the same order as the coordinates of $\br$. So we may use the `composition code' decomposition technique from~\cite[Section~IV.B]{ImminkWeber14} to decode more efficiently, only storing codewords that are in sorted order. Finally, we observe that $2$-constrained codes $C$ have the property that whenever $\bx\in C$ then its \emph{complement} $\by=(q-1)\bo-\bx$ also lies in $C$. We note that $\zeta(\bx)=-\zeta(\by)$ and so we find that $\sigma_\bx=\sigma_{\by}$ and $\rho_{\br,\bx}=-\rho_{\br,\by}$ for any non-constant received word $\br$. So for codes which are closed under taking complements, we only need to store one codeword from each pair $\{\bx,\by\}$. If we do this, we search for a codeword $\bh$ that minimises $\sigma_{\bh}^2(1-\rho_{\br,\bh}^2)$; we then decode to the complement of $\bh$ when $\rho_{\br,\bh}<0$ and decode to $\bh$ otherwise. This technique can be combined with the composition code technique above. The technique can also be used with the decoder in~\cite{ImminkWeber14}: here we find a codeword maximising $|\rho_{\br,\bh}|$, and decode to this codeword if $\rho_{\br,\bh}\geq 0$ or to its complement otherwise.

\section{The distance between codewords}
\label{sec:distance}

For codewords $\bu,\by\in C$, we define a (squared) distance measure $\delta'(\bu,\bv)$ by
\[
\delta'(\bu,\bv)=\begin{cases}
\sigma_\bu^2\sigma_\bv^2(1-\rho_{\bu,\bv}^2)/\sigma_{\bu+\bv}^2&\text{ when }
\rho_{\bu,\bv}>-\min\{\sigma_\bv/\sigma_{\bu},\sigma_\bu/\sigma_{\bv}\},\\
\min\{\sigma_\bu^2,\sigma_\bv^2\}&\text{ otherwise}.
\end{cases}
\]
Note that $\delta'(\bu,\bv)=\delta'(\bv,\bu)$. Also note that
$\delta'(\bu,\bv)$ depends only on $\zeta(\bu)$ and $\zeta(\bv)$, by~\eqref{eqn:sigma_geometric} and \eqref{eqn:rho_inner_product}. Finally, we claim that
\begin{equation}
\label{eqn:dist_bound}
\sigma_\bu^2\sigma_\bv^2(1-\rho_{\bu,\bv}^2)/\sigma_{\bu+\bv}^2\leq \min\{\su^2,\sv^2\}.
\end{equation}
To see this, we may verify by routine calculation that
\begin{align}
\sigma_\bu^2\sigma_\bv^2(1-\rho_{\bu,\bv}^2)/\sigma_{\bu+\bv}^2
&=\frac{\langle \zeta(\bu),\zeta(\bu)\rangle\langle \zeta(\bv),\zeta(\bv)\rangle-\langle \zeta(\bu),\zeta(\bv)\rangle^2}{\langle \zeta(\bu+\bv),\zeta(\bu+\bv)\rangle}\nonumber\\
\label{eqn:rearrange}
&=\langle \zeta(\bu),\zeta(\bu)\rangle-\frac{\langle \zeta(\bu),\zeta(\bu)+\zeta(\bv))\rangle^2}{\langle \zeta(\bu)+\zeta(\bv),\zeta(\bu)+\zeta(\bv)\rangle}\\
&\leq \langle \zeta(\bu),\zeta(\bu)\rangle=\su^2.\nonumber
\end{align}
A similar calculation shows that the left hand side of~\eqref{eqn:dist_bound} is at most $\sv^2$, and so the claim follows.

In this section, we will give a geometric interpretation for $\delta'(\bu,\bv)$, and we relate the minimum distance of a code (using this notion of distance) to the error rate of a maximum likelihood decoder.

We note that~\cite{ImminkWeber14} defines a different distance measure (namely the distance $d_2(\bu,\bv)=2\sigma_\bu^2(1-\rho_{\bu,\bv})$, which is not symmetrical in $\bu$ and $\bv$) to be used to calculate the minimum distance of a code in this context. This distance measure is natural for the decoder in~\cite{ImminkWeber14}; in Section~\ref{sec:comparison} we briefly compare this measure with the measure above. 

\begin{lemma}
\label{lem:ray}
Let $\bx,\by\in \bR^n$. Then $\delta'=\delta'(\bx,\by)$ is the largest real number $\delta'$ with the following property. Let $B(\bx,\delta')$ be the ball in $Z$ of radius $\sqrt{\delta'}$ and centre $\zeta(\bx)$ (using Euclidean distance). Let $B(\by,\delta')$ be the ball in $Z$ of radius $\sqrt{\delta'}$ and centre $\zeta(\by)$. Then there is no ray $R_{\zeta(\br)}$ that intersects the interior of both $B(\bx,\delta')$ and $B(\by,\delta')$.
\end{lemma}

\begin{figure}
\begin{center}
\includegraphics[width=80mm]{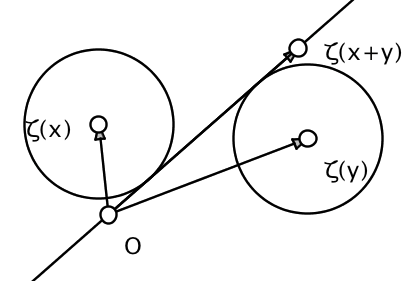}
\end{center}
\caption{A hyperplane in $Z$ at equal distance from $\zeta(\bx)$ and $\zeta(\by)$}
\label{fig:generic_min_radius}
\end{figure}

\begin{proof}
Firstly, suppose that $\rho_{\bx,\by}>-\min\{\sigma_\by/\sigma_{\bx},\sigma_\bx/\sigma_{\by}\}$. In particular this means that $\zeta(\bx)\not=-\zeta(\by)$, and so $\zeta(\bx+\by)$ is a non-zero vector.

The typical situation in this case is drawn in Figure~\ref{fig:generic_min_radius}.
Let $P$ be the subplane of $Z$ generated by $\zeta(\bx)$ and $\zeta(\by)$, and let $K$ be the subplane of vectors orthogonal to $P$. Let $H$ be the hyperplane in $Z$ generated by $\zeta(\bx+\by)$ and $K$. We have that $\zeta(\bx)$ and $\zeta(\by)$ lie on different sides of $H$. The closest point in $H$ to $\zeta(\bx)$ lies in $P$, and so lies on the line generated by $\zeta(\bx)+\zeta(\by)$; the same is true for the closest point to $\zeta(\by)$. Setting $\theta$ to be the angle between  $\zeta(\bx)$ and $\zeta(\bx)+\zeta(\by)$, we find that the squared distance between $H$ and $\zeta(\bx)$ is
\begin{align*}
||\zeta(\bx)||^2\sin^2\theta&=||\zeta(\bx)||^2-||\zeta(\bx)||^2\cos^2\theta\\
&=\langle\zeta(\bx),\zeta(\bx)\rangle-\frac{\langle\zeta(\bx),\zeta(\bx)+\zeta(\by)\rangle^2}{\langle\zeta(\bx)+\zeta(\by),\zeta(\bx)+\zeta(\by)}\\
&=\sigma_\bx^2\sigma_\by^2(1-\rho^2_{\bx,\by})/\sigma^2_{\bx+\by} \text{ by~\eqref{eqn:rearrange}}\\
&=\delta'.
\end{align*}
So the interior of $B(\bx,\delta')$ does not intersect $H$. Similarly, $\zeta(\by)$ is also at distance $\delta'$ from $H$ and so the interior of $B(\by,\delta')$ does not intersect $H$. So the interiors of $B(\bx,\delta')$  and $B(\by,\delta')$ lie on different sides of a hyperplane, and therefore no ray from the origin intersects them both, as required.

We now show that the value for $\delta'$ is optimal, by proving that the ray $R_{\zeta(\bx+\by)}$ touches the boundaries of both $B(\bx,\delta')$ and $B(\by,\delta')$. The nearest point to $\zeta(\bx)$ on the line generated by $\zeta(\bx+\by)$ is given by
\[
\frac{||\zeta(\bx)||\cos\theta}{||\zeta(\bx+\by)||}\zeta(\bx+\by)=\frac{\langle\zeta(\bx),\zeta(\bx+\by)\rangle}{||\zeta(\bx+\by)||^2}\zeta(\bx+\by).
\]
So $R_{\zeta(\bx+\by)}$ touches $B(\bx,\delta')$ if and only if $\langle\zeta(\bx),\zeta(\bx+\by)\rangle>0$. But
\begin{align*}
\langle \zeta(\bx),\zeta(\bx+\by)\rangle&=\langle \zeta(\bx),\zeta(\bx)\rangle+\langle \zeta(\bx),\zeta(\by)\rangle\\
&=\sigma_\bx^2+\sigma_{\bx}\sigma_{\by}\rho_{\bx,\by}\text{ by~\eqref{eqn:sigma_geometric} and~\eqref{eqn:rho_inner_product}}\\
&>\sigma_\bx^2-\sigma_{\bx}\sigma_{\by}(\sigma_{\bx}/\sigma_{\by})\\
&=0.
\end{align*}
So $R_{\zeta(\bx+\by)}$ touches $B(\bx,\delta')$. The argument that $R_{\zeta(\bx+\by)}$ touches $B(\by,\delta')$ is similar, and uses the fact that $\rho_{\bx,\by}>-\sigma_{\by}/\sigma_{\bx}$. This shows that our value for $\delta'$ is optimal in this case.

We now turn to the case when $\rho_{\bx,\by}\leq-\min\{\sigma_\by/\sigma_{\bx},\sigma_\bx/\sigma_{\by}\}$. See Figure~\ref{fig:bigger_min_radius} for a typical situation. Without loss of generality, assume that $\sigma_\bx\leq \sigma_\by$. So $\delta'=\sigma_\bx^2$ and $\rho_{\bx,\by}\leq-\sigma_\bx/\sigma_\by$.

\begin{figure}
\begin{center}
\includegraphics[width=80mm]{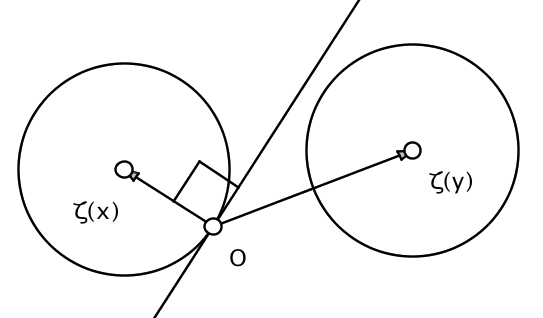}
\end{center}
\caption{A typical case when $||\zeta(\bx)||\leq||\zeta(\by)||$}
\label{fig:bigger_min_radius}
\end{figure}

Let $H=\zeta(\bx)^{\perp}$, so
\[
H=\{\bu\in Z:\langle\zeta(\bx),\bu\rangle=0\}
\]
is the hyperplane in $Z$ of all vectors that are orthogonal to $\bx$. Clearly the nearest point on $H$ to $\zeta(\bx)$ is the origin, so $\zeta(\bx)$ is at distance $||\zeta(\bx)||=\sigma_{\bx}$ from $H$. Moreover, all points $\bu$ in the interior of $B(\bx,\delta')$ have $\langle\zeta(\bx),\bu\rangle>0$. Now let $\bu$ be a point in the interior of $B(\by,\delta')$. Then $\bu=\zeta(\by)+\bv$, where $||\bv||<\sqrt{\delta'}$ and so
\begin{align*}
\langle \zeta(\bx),\bu\rangle&=\langle \zeta(\bx),\zeta(\by)\rangle+\langle \zeta(\bx),\bv\rangle\\
&<\langle \zeta(\bx),\zeta(\by)\rangle+||\zeta(\bx)||\sqrt{\delta'}\\
&\leq||\zeta(\bx)||\,||\zeta(\by)||\rho_{\bx,\by}+||\zeta(\bx)||\sx\\
&\leq \sx\sigma_{\by}(-\sigma_{\bx}/\sigma_{\by})+\sx^2\\
&=0.
\end{align*}
Thus all points in the interior of $B(\by,\delta')$ lie on the opposite side of the hyperplane $H$ to the points in the interior of $B(\bx,\delta')$. So no ray from the origin can pass through both  $B(\bx,\delta')$ and  $B(\by,\delta')$, as required. Finally, it is easy to see that no larger value of $\delta'$ can have this property, for when $\delta'>\sigma_\bx^2$ we find that the origin is in the interior of $B(\bx,\delta')$, and so all rays from the origin (including, for example, $R_{\zeta(y)}$) pass through $B(\bx,\delta')$.
\end{proof}

\begin{theorem}
\label{thm:decoder_probability}
Let $C\subseteq \bR^n$ be a finite set of non-constant codewords. Define the minimum distance $\delta'$ of $C$ by
\[
\delta'=\min_{\bx,\by\in C,\bx\not=\by}\delta'(\bx,\by).
\]
The word error probability of a maximum likelihood decoder is bounded above by the probability that $\chi^2(n-1)\geq \delta'/\sigma^2$, where $\chi^2(n-1)$ is the chi-squared distribution with $n-1$ degrees of freedom.
\end{theorem}
\begin{proof}
When a codeword $\bx$ is transmitted, the decoder receives a vector $\br=a(\bx+\bnu)+b$ where $a$ and $b$ are unknown, and the components of $\bnu$ are normally distributed with mean $0$ and standard deviation $\sigma$. Let $\be_1,\be_2,\ldots,\be_n$ be an orthonormal basis for $\bR^n$, with  $\be_1,\be_2,\ldots,\be_{n-1}$ spanning $Z$. We may write $\bnu=\bnu'+c\be_n$ where
\[
\bnu'=\nu'_1\be_1+\nu'_2\be_2+\cdots+\nu'_{n-1}\be_{n-1}
\]
and where the real numbers $\nu'_i$ and $c$ are independent and normally distributed with mean $0$ and standard deviation $\sigma$. Now $||\bnu'||^2/\sigma^2$ is a chi-squared random variable with $n-1$ degrees of freedom, so $||\bnu'||^2<\delta'$ with probability equal to the probability that $\chi^2(n-1)\geq \delta'/\sigma^2$. Assume that $||\bnu'||^2<\delta'$. It suffices to show that our maximum likelihood decoder returns the codeword $\bx$.

Note that $\zeta(\bnu)=\bnu'$, and so $\zeta(\br)=a(\zeta(\bx)+\bnu')$. The ray $R_{\zeta(\br)}$ passes within a squared distance of $||\bnu'||^2$ from $\zeta(\bx)$, since $\zeta(\bx)+\bnu'$ lies on this ray. So $\ell_\bx(\br)\leq ||\bnu'||^2<\delta'$. Let $\bh\in C$ be such that $\bh\not=\bx$. Lemma~\ref{lem:ray} and the definition of $\delta'$ shows that $R_{\zeta(\br)}$ cannot intersect the interior of a ball in $Z$ of radius $\delta'$ centred at $\zeta(\bh)$. Hence $\ell_\bh(\br)\geq \delta'$. Thus a maximum likelihood decoder will correctly decode to $\bx$.
\end{proof}

\section{Comparing the two decoders}
\label{sec:comparison}

Simulations indicate that the decoder in~\cite{ImminkWeber14} has a comparable performance with the maximum likelihood decoder when word error rate is considered. Figure~\ref{fig:WER4} shows the results of a simulation for the maximum likelihood decoder when $q=2$ and $n=4$ for a range of noise levels when a $1$-constrained code~\cite{ImminkWeber14} is used: the horizontal axis is the signal to noise ratio, defined as $-20\log_{10} \sigma$, and the vertical axis is the word error rate. Each point was the result of 10,000 trials with $a=1.07$ and $b=0.07$. Figure~\ref{fig:WER12} gives a similar situation when $n=12$. In each figure, simulation results are plotted along with the error rate predicted by averaging the bound in Theorem~\ref{thm:decoder_probability} over all subcodes of size $2$. These parameters are chosen for direct comparison with Figure~5 in~\cite{ImminkWeber14}.

\begin{figure}
\begin{center}
\includegraphics[width=80mm]{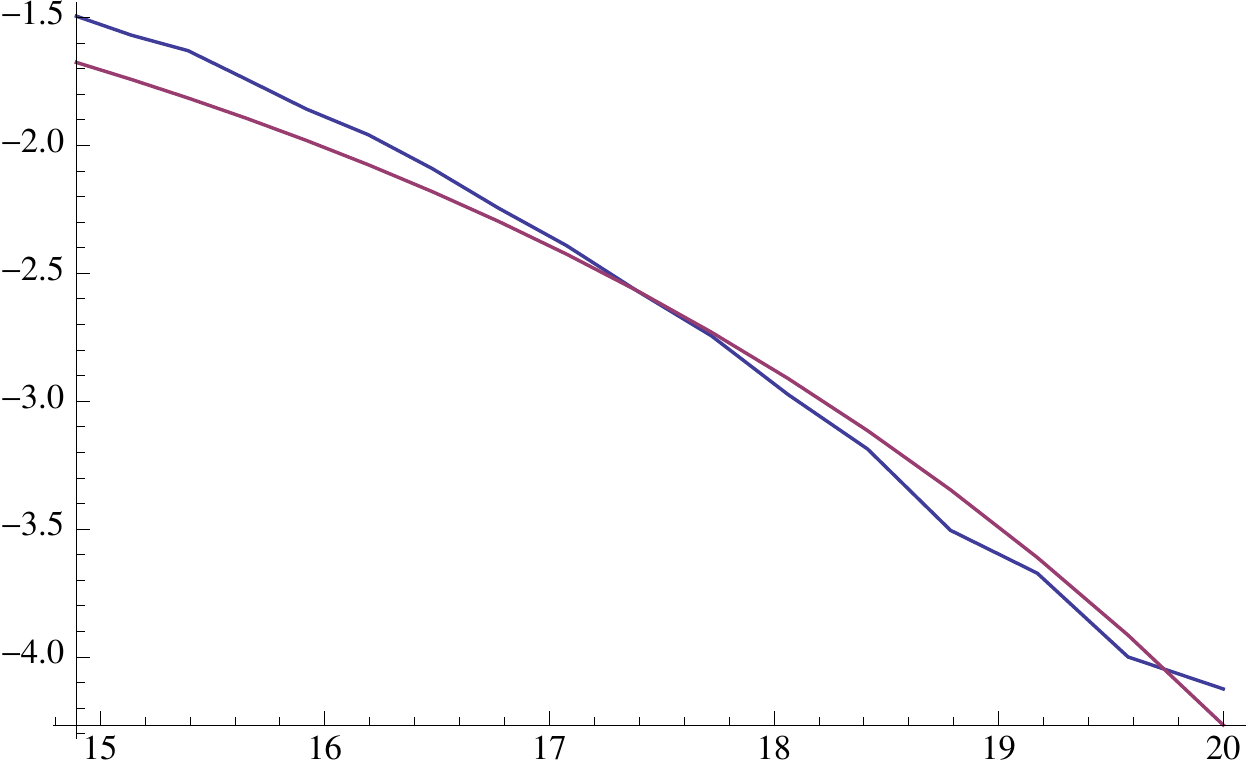}
\end{center}
\caption{Word error rate when $q=2$ and $n=4$}
\label{fig:WER4}
\end{figure}

\begin{figure}
\begin{center}
\includegraphics[width=80mm]{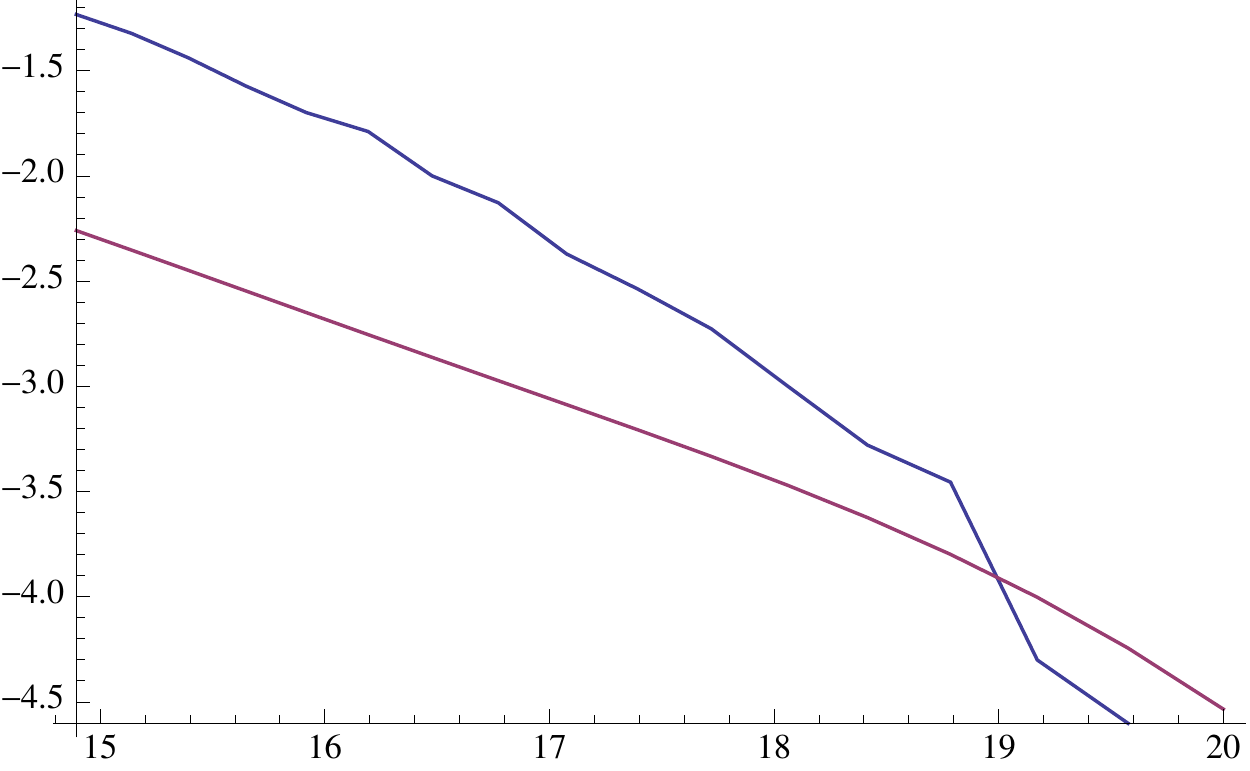}
\end{center}
\caption{Word error rate when $q=2$ and $n=12$}
\label{fig:WER12}
\end{figure}

Figure~\ref{fig:distance_comparison} is a scatter plot of two notions of distance for 10,000 random vectors when $q=100$ and $n=20$: the distance $\delta'(\bu,\bv)$ defined in Section~\ref{sec:distance} and the distance $d_2(\bu,\bv)$ defined in Section~IV.B of~\cite{ImminkWeber14} for the purposes of estimating word error rates. The figure shows a close to linear relationship between these two quantities for random vectors. Figure~\ref{fig:likelihood_comparison} is a scatter plot of two likelihood functions (namely Pearson distance and the likelihood function $\ell_\bx(\by)$ used by the maximum likelihood decoder) for a similarly randomly generated collection of vectors. Again, a close to linear relationship can be observed, which provides an explanation for the similar performance of the corresponding decoders.

\begin{figure}
\begin{center}
\includegraphics[width=80mm]{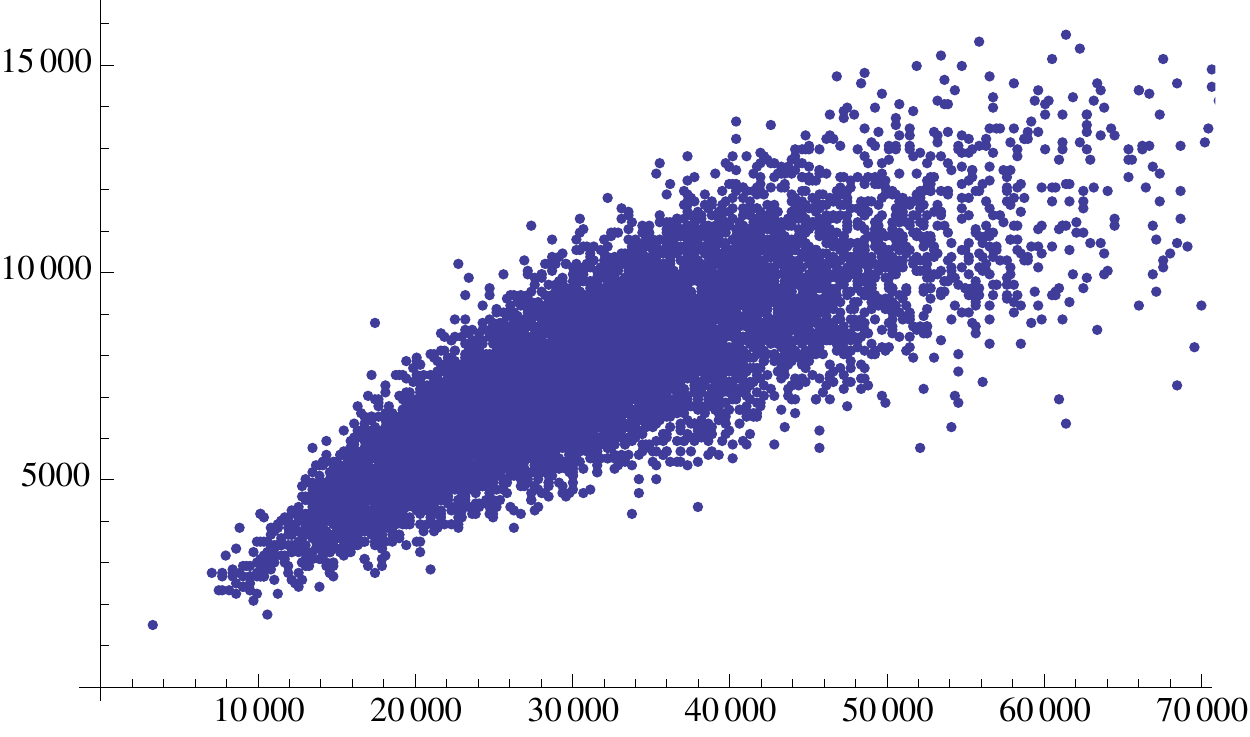}
\end{center}
\caption{Two distances: $d_2$ against $\delta'$}
\label{fig:distance_comparison}
\end{figure}

\begin{figure}
\begin{center}
\includegraphics[width=80mm]{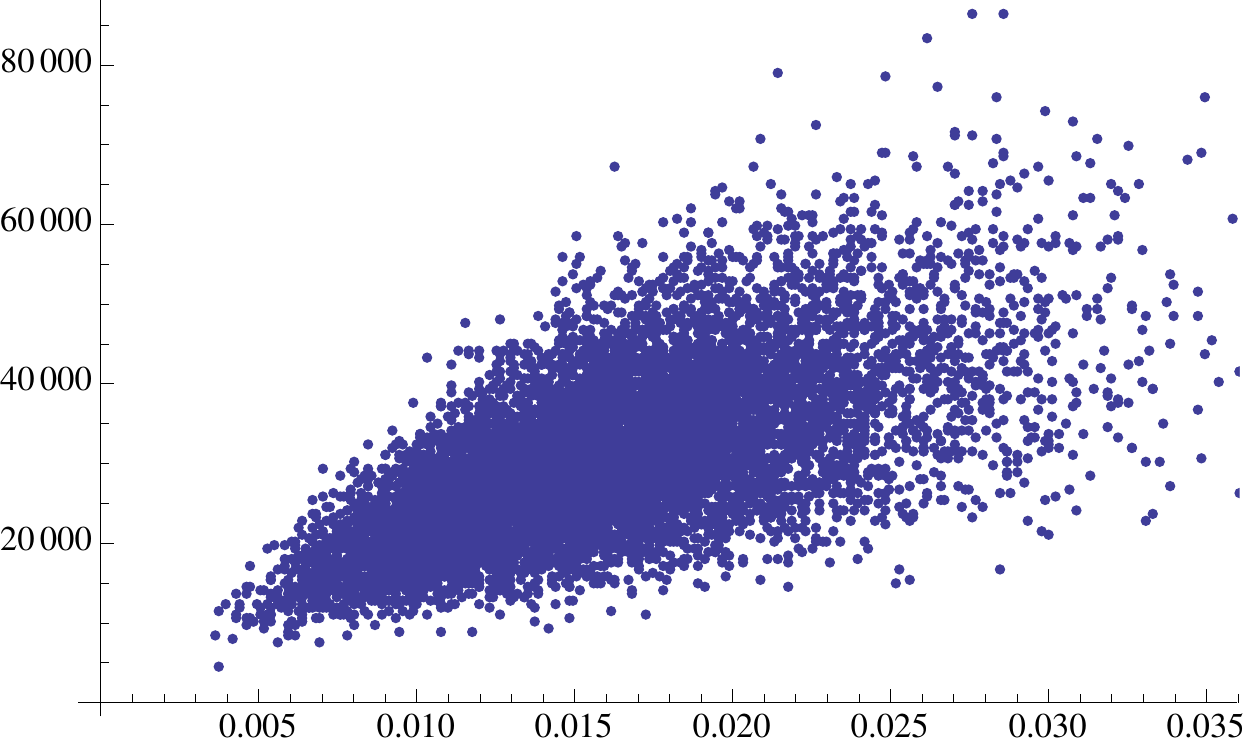}
\end{center}
\caption{Two likelihood functions: Pearson distance against $\ell_\bx(\by)$}
\label{fig:likelihood_comparison}
\end{figure}

\section{Comments}
\label{sec:comments}

\subsection{A geometric meaning for Pearson distance} Since the offset $b$ is arbitrary and unknown, and changes the mean of a vector by $b$, it seems sensible to normalise codewords and received words to have mean $0$. In other words, we consider the words $\zeta(\bh)=\bh-\overline{\bh}\bo$ and $\zeta(\br)=\br-\overline{\br}\bo$ rather than $\bh$ and $\br$. Scaling a vector of mean $0$ by $a$ does not change the mean, but scales the standard deviation by a factor of $a$. So it seems sensible to scale our normalised vectors so that they have standard deviation $1$: if our original vectors were non-constant, we can always find a scaling factor $a$ that does this. The resulting vectors, $(\bh-\overline{\bh}\bo)/\sh$ and $(\br-\overline{\br}\bo)/\sr$, lie on an $n-1$-dimensional unit sphere, centred at the origin. A natural distance measure between two vectors $\bu$ and $\bv$ on this unit sphere is their squared Euclidean distance, and it is not difficult to show that this is exactly $\delta_{\mathrm{Pearson}}(\bu,\bv)$. 

\subsection{Why are constant codewords forbidden?} The (unnormalised) standard deviation of a constant codeword is $0$, so the Pearson correlation coefficient $\rho_{\br,\bh}$ is not defined when $\bh$ is constant. But the forbidding of constant codewords is an artifact of the channel itself, not just the distance measure that is proposed for decoding. To see this, suppose that $\bh=\alpha \bo$ is a codeword. Let $\br$ be a received word, and define $\bs=\br-\bh$. For any positive $\epsilon\in \mathbb{R}$ we have
\begin{align*}
\epsilon^{-1}(\bh+\epsilon\bs)+(1-\epsilon^{-1})\alpha\bo&=
\bs+\alpha\bo=\br.
\end{align*}
Setting $a=\epsilon^{-1}$, $b=(1-\epsilon^{-1})\alpha$ and $\bnu=\epsilon\bs$ we have that $\br=a(\bh+\bnu)+b\bo$. But $\epsilon$ may be taken to be arbitrarily small, and so we see $\br$ could have been received when $\bh$ was transmitted, with an abitrarily small error vector $\bnu=\epsilon\bs$. So any reasonable decoder for this channel would decode \emph{every} received vector to $\bh$, and a sensible distance measure would set the distance between $\bh$ and any other vector as $0$.

\subsection{Future work} Weber, Immink and Blackburn~\cite{ImminkWeberPearson} have studied optimal Pearson codes, which are the largest codes contained in $\{0,q-1\}^n$ that can be correctly decoded in the zero-error case (when $\sigma=0$, and so $\nu=\mathbf{0}$). It would be very interesting to fully explore the interplay between the error correcting capacity of codes when $\sigma>0$ and the rate of an optimal code. We hope that the distance between codewords that is defined in Section~\ref{sec:distance} will provide a tool to accomplish this. 

\paragraph{Acknowledgements} The author would like to thank Alexey Koloydenko for fruitful discussions on maximum likelihood estimation, and Kees S. Immink and Jos Weber for commenting on an earlier draft of the manuscript. The author would also like to acknowledge the help of two software packages that were used to conduct experiments and simulations: Compass and Ruler~\cite{Grothmann}, and Mathematica~\cite{Wolfram}.

\end{document}